\documentclass{article}

\usepackage[margin=1in]{geometry}
\usepackage{amsmath,amssymb,amsthm,enumerate,url,siunitx,multirow,hhline,graphicx,cancel,soul}
\usepackage{unicode-math}

\newtheorem{theorem}{Theorem}

\newtheorem{lemma}[theorem]{Lemma}
\newtheorem{corollary}[theorem]{Corollary}

\newtheorem{remark}[theorem]{Remark}

\newcommand{\paren}[1]{\left(#1\right)}

\renewcommand{\brace}[1]{\left\{#1\right\}}
\renewcommand{\ang}[1]{\left\langle#1\right\rangle}

\newcommand{\ceil}[1]{\left\lceil#1\right\rceil}

\newcommand{\F}{{\mathbb{F}}}
\newcommand{\A}{{\mathcal{A}}}
\newcommand{\B}{{\mathcal{B}}}
\newcommand{\diag}[1]{\mathrm{diag}\paren{#1}}
\renewcommand{\ker}[1]{\mathrm{ker}\paren{#1}}

\newcommand{\rk}[1]{\mathrm{rk}\paren{#1}}
\newcommand{\M}[1]{\begin{bmatrix}#1\end{bmatrix}}

\renewcommand{\span}[1]{\mathrm{span}\paren{#1}}
\renewcommand{\dim}[1]{\mathrm{dim}\paren{#1}}

\newcommand{\elt}{\mathrm{elt}}
\newcommand{\im}[1]{\mathrm{im}\paren{#1}}

\let\oldtextbf=\textbf
\renewcommand\textbf[1]{{\boldmath\oldtextbf{#1}}}

\newcommand{\COMMENT}[1]{}

\title{New lower bounds on tensor rank of $(n,2,m)$ matrix multiplication with GPT-6}
\author{Jason Yang}
\date{September 2026}

\begin{document}

\maketitle

\begin{abstract}
The tensor rank of $n\times 2$-with-$2\times m$ matrix multiplication is
at least $\left\lceil \frac{24m+2}{5} \right\rceil$ if $m\ge n=3$
and at least $(n+2)m+n-3$ if $m\ge n\ge 4$, over arbitrary fields.
This result matches the Hopcroft-Kerr upper bound for $n=3$ and $m\le 6$.
\end{abstract}

\section{Introduction}
An $R$-rank decomposition of the $\ang{k,l,m}$ matrix multiplication tensor is a list of linear maps $\alpha_r:\F^{k\times l}\to \F,\ \beta:\F^{l\times m}\to \F$ and matrices $C_r\in\F^{k\times m}$ for $0\le r<R$
such that for all $A\in\F^{k\times l}$ and $B\in\F^{l\times m}$,
\[AB=\sum_r \alpha_r(A)\beta_r(B)C_r.\]


The rank of $\ang{k,l,m}$, denoted as $R(\ang{k,l,m})$, is the smallest possible $R$ such that a decomposition exists for $\ang{k,l,m}$.
Computing this quantity for various $\ang{k,l,m}$ is the central task fast matrix multiplication; see \cite{survey} for more information.
Permuting $k$, $l$, and $m$ does not change the rank of $\ang{k,l,m}$; however, changing the ground field $\F$ might.

Our focus is on lower bounds for $R(\ang{n,2,m})$ that apply to all fields.
Recently, Shitov \cite{shitov} used GPT-5.6 to find a simple proof that $R(\ang{2,2,m})\ge \frac{7}{2}m$.
This quantity matches the best-known upper bound $R(\ang{n,2,m})\le \ceil{\frac{3nm+\max(n,m)}{2}}$ from Hopcroft-Kerr \cite{hopcroft}.

Building off of Shitov, we use GPT-6 to prove the following:

\begin{theorem}
\label{32m}
For $m\ge 3$,
$R(\ang{3,2,m})\ge \ceil{\frac{24m+2}{5}}$.
\end{theorem}

\begin{theorem}
\label{n2m}
For $m\ge n\ge 4$, $R(\ang{n,2,m})\ge (n+2)m+(n-3)$.
\end{theorem}

For $n,m\ge 3$, these lower bounds are strict improvement over a special case of a result from Bl\"aser \cite{blaser99}:
\[\forall 2\le k\le m:\ R(\ang{k,l,m})\ge kl+lm+k-l+m-3\]
\[\Rightarrow R(\ang{2,n,m})\ge
nm+n+m-1\]

Theorem \ref{32m} matches Hopcroft-Kerr for $m\le 6$, proving $R(\ang{3,2,3})=15,\ R(\ang{3,2,4})=20,\ R(\ang{3,2,5})=25$, and $R(\ang{3,2,6})=30$. Previously, Burichenko \cite{proof323} gave a much longer proof that $R(\ang{3,2,3})=15$.
Additionally, Theorem \ref{n2m} proves $R(\ang{4,2,4})\ge 25$ and $R(\ang{4,2,5})\ge 31$, which are each one away from the current upper bounds \cite{hopcroft} \cite{novikov}.

An early version of this paper proved the weaker result $R(\ang{3,2,m})\ge \ceil{\frac{24m}{5}}$, which shortly after was improved to $R(\ang{3,2,m})> \ceil{\frac{24m}{5}}$ by Tsyganov et. al. \cite{tsyganov}. By integrality of rank, this is equivalent to $R(\ang{3,2,m})\ge\frac{24m+1}{5}$.

\section{Previous work}
For completeness, we restate the proof from \cite{shitov} that $R(\ang{2,2,m})\ge \frac{7}{2}m$, with slight simplification.

The decomposition equation $AB=\sum_r \alpha_r(A)\beta_r(B)C_r$ be rewritten as
\[AB=WD(A)\B(B) \ \forall A,B;\]
\[D(A):=\M{\alpha_0(A)\\&\ddots\\&&\alpha_{R-1}(A)},\ 
\B(B):=\M{\beta_0(B)\\\vdots\\\beta_{R-1}(B)},\ 
W(t\in \F^R):=\sum_r t_rC_r.\]

We immediately have that $\B$ is injective, since $\B(B)=0$ implies $AB=0$ for all $A$, forcing $B=0$.
Additionally, $W$ is surjective because $AB$ covers all possible matrices of its shape as $A$ and $B$ are varied.
By the rank-nullity theorem, $\dim{\ker{W}}=r-2m$.

Without loss of generality, assume
\begin{itemize}
    \item $\F$ is infinite (otherwise we can replace $\F$ with its algebraic closure);
    \item no $\alpha_r$ is constant 0 (otherwise we can remove a rank of the decomposition).
\end{itemize}

Then a generic rank-1 matrix $M$ satisfies $\alpha_r(M)\ne 0$ for all $r$.
We can apply a change of basis $(\alpha'_r(A),\beta'_r(B),C'_r):=(\alpha_r(S^{-1}AT),\ \beta_r(T^{-1}B),\ S^{-1}C_r)$
such that $\alpha_r(e_{0,0})\ne 0$ for all $r$,
where $e_{i,j}$ denotes the 1-hot matrix with a 1 at $(i,j)$ and whose shape is inferred from context. Thus, $D(e_{0,0})$ is invertible.

Define the following:
\begin{itemize}
    \item $U=\M{0&\cdots&0\\ *&\cdots&*}$ is the space of $2\times m$ matrices where the top row is 0
    \item $f_i(x)=D(e_{i,0})\B(x)$
    \item $g(x)=D(e_{0,1})\B(x)-D(e_{0,0})\B(e_{0,1}x)$
    \item $J=f_0(U)$
\end{itemize}

It is clear that $f_i(U),g(U)\subseteq \ker{W}$, $f_0$ is injective, and $\dim{J}=\dim{U}=m$.

Now consider the map
\[\Phi:U\to (\ker{W}/J)^2;\
\Phi(u):=(f_1(u),\ g(u))+_\elt J,\]

where $+_\elt$ represents element-wise addition, i.e. $\Phi(u):=(f_1(u)+J,\ g(u)+J)$.

\begin{lemma}
$\Phi$ is injective.
\end{lemma}
\begin{proof}


$\Phi(u)=0$ means $f_1(u), g(u)\in J$, so
\[D(e_{0,1})f_1(u)-D(e_{1,0})g(u)\in D(e_{0,1})J-D(e_{1,0})J.\]
The left hand side simplifies to
\[\cancel{D(e_{0,1})D(e_{1,0})\B(u)}
-\cancel{D(e_{1,0})D(e_{0,1})\B(u)} + D(e_{1,0})D(e_{0,0})\B(e_{0,1}u).\]

We can factor out $D(e_{0,0})$ from both sides because it is invertible, so
\[D(e_{1,0})\B(e_{0,1}u)\in D(e_{0,1})\B(U)-D(e_{1,0})\B(U)\]

Left-multiplying by $W$ yields
\[e_{1,0}e_{0,1}u\in e_{0,1}U-\cancel{e_{1,0}U}\]
\[e_{1,1}u\in e_{0,1}U.\]
Since the left hand side only has nonzeros on row 1, but the right hand side only has nonzeros on row 0, $u=0$.
\end{proof}

Since $\Phi(U)\subseteq (\ker{W}/J)^2$, comparing dimensions yields $m\le 2((r-2m)-m)
\Rightarrow r\ge \frac{7}{2}m$.

\begin{remark}
Since $f_0$ is injective, one can avoid using quotient subspaces by showing that the map $\phi(u,v,v')=(f_1(u)+f_0(v),\ g(u)+f_0(v'))$ is injective. Then $\phi(U^3)\subseteq \ker{W}^2$, yielding the same lower bound.
\end{remark}

\begin{remark}
Using $f_0$ alone would yield
$f_0(U)\subseteq W
\Rightarrow m\le r-2m \Rightarrow r\ge 3m$.
\end{remark}

\section{Strong normalization}
To obtain our tensor rank lower bounds, we first enforce a stronger change-of-basis normalization on decompositions that works for general tensors $\ang{k,l,m}$. The maps $f_i(x)=D(e_{i,0})\B(x)$ are defined as before.

\begin{lemma}
For a decomposition $(W,D,\B)$ of $\ang{k,l,m}$, we can change basis such that
\[\forall v\in \F^{l\times m},\ v\ne 0:\ \brace{f_i(v):0\le i<k} \textrm{ is linearly independent}\]
and $D(e_{0,0})$ is invertible.
\end{lemma}
\begin{proof}
We show that generic $u\in\F^{k\times 1},\ p\in\F^{1\times l}$ satisfy
\begin{itemize}
    \item $D(up)$ is invertible;
    \item $\forall v\in \F^{l\times m},\ v\ne 0:$ $(a\mapsto D(ap)\B(v))$ is injective.
\end{itemize}

Afterwards, we can select linearly independent $u_0,u_1,\dots,u_{k-1}\in \F^{k\times 1}$ such that $u_0=u$ and apply a change of basis that maps $u_i p\mapsto e_{i,0}$.

To show that $D(up)$ is invertible, note that all entries along its diagonal are linear with respect to $up$, but the matrix $up$ is not confined to any strict linear subspace of $\F^{k\times l}$.

For any fixed $v$, the map $(a\mapsto D(ap)\B(v))$ is injective for any $p$ such that $pv\ne 0$, as having $D(ap)\B(v)=0$ would imply $WD(ap)\B(v)=a(pv)=0 \Rightarrow a=0$.

Let $\A(a)=\diag{D(a)}$ and $S=\brace{i:\B(v)_i\ne 0}$: then having $(a\mapsto D(ap)\B(v))$ be injective is equivalent to $(a\mapsto \A(ap)_S)$ being injective.
Since there exists $p$ making the latter injective, $|S|\ge k$ and\footnote{Another proof that $|S|\ge k$ is that for any nonzero $v\in \F^{l\times m}$, the product $uv=WD(u)\B(v)$ spans dimension $\ge n$ as we vary $u\in\F^{k\times l}$, with equality if $v$ has matrix rank 1.}
the matrix corresponding to $(a\mapsto \A(ap)_S)$ contains at least one $k\times k$ minor that is a non-constant polynomial in $p$.
Since there are finitely many distinct $S$ across all $v$, $p$ only has to force finitely many non-constant polynomials to be simultaneously nonzero, which a generic $p$ will do.
\end{proof}

\begin{corollary}
\label{lin-combo-injective}
For any $\gamma\in\F^k,\ \gamma\ne 0$, the map $v\mapsto \sum_i \gamma_i f_i(v)$ is injective.
\end{corollary}

\begin{lemma}
\label{subspace-expand}
For any nonzero subspace $S\subseteq\F^{l\times m}$, $\dim{f_0(S)+\dots+f_{k-1}(S)}\ge \dim{S}+(k-1)$.
\end{lemma}
\begin{proof}
Define $P_i=D(e_{i,0})D(e_{0,0})^{-1}$:
then $f_i(v)=P_i f_0(v)$.
Because all $P_i$ are diagonal, they simultaneously preserve the subspaces $0=E_0\subset \dots \subset E_r=\F^r$, where $E_i$ consists of all vectors whose elements after the first $i$ many are 0.

Choose the smallest $j$ such that the intersection $X:=f_0(S)\cap E_j$ is nonzero; then $\dim{X}=1$, so $X=\span{\brace{f_0(v)}}$ for some nonzero $v$. Then $T:=\span{\brace{f_i(v):0\le i<k}}$ has dimension $k$.
Furthermore, $T\subseteq E_j$, due to invariance of $E_j$ under all $P_i$, forcing $f_0(S)\cap T=X$.
Thus, $\dim{f_0(S)+T}\ge \dim{f_0(S)}+k-1$.
Since $f_0$ is injective, $\dim{f_0(S)}=\dim{S}$.
\end{proof}

\begin{lemma}
\label{invariant-subspace}
For $\F$ algebraically closed, there is no pair $(S,i)$ for a nonzero subspace $S\subseteq \F^{l\times m}$ and index $i\ne 0$ such that $f_i(S)\subseteq f_0(S)$.
\end{lemma}
\begin{proof}
Suppose there is such $(S,i)$.
Let $V=f_0(S)$ and $P=D(e_{i,0})D(e_{0,0})^{-1}$:
then $f_i(S)=PV\subseteq V$, so there must exist some nonzero $v\in V$ that is an eigenvector\footnote{This is because $P$ is equivalent to a (finite-sized) square matrix when viewed as a map restricted to $V$.} of $P$. But then $f_0(v)$ and $f_i(v)$ are parallel, contradicting strong normalization.
\end{proof}

\section{New lower bounds}
For the $\ang{n,2,m}$ tensor with $m,n\ge 3$, define $U, f_i, g, J$ similarly as in $\ang{2,2,m}$.
Then $\dim{U}=m$, $\dim{\ker{W}}=r-nm$, and $\dim{\ker{W}/J}=r-(n+1)m$.

Construct
\begin{itemize}
    \item $\ell:U^{n-1}\to \ker{W}:\ \ell(u_1,\dots,u_{n-1})=\sum_{1\le i<n} f_i(u_i)$

    \item $U^{\wedge d}:=\brace{c\in U^{d\times d}:\ c_{i,i}=0,\  c_{j,i}=-c_{i,j} \ \forall 1\le i,j\le d}$

    \item $h(c\in U^{\wedge n-1})=\paren{\ell(c_{1,:}),\dots,\ell(c_{n-1,:})}$,
    where $c_{i,:}$ denotes the $i$-th row of $c$

    We call $h$ the \textit{skew augmentation} of the map $\ell$.
\end{itemize}

For convenience, define the quotient maps $L(u):=\ell(u)+J,\ H(c):=h(c)+_\elt J$. The elements of $h$ and $H$ are indexed $1,\dots,n-1$.

Our new map of interest is
\[\Phi:U^{n-1}\oplus U^{\wedge (n-1)} \to \im{L}\oplus (\ker{W}/J)^{n-1}\]
\[\Phi(u,c)
=\paren{
\ell(u),\ 
\dots,
-g(u_i)+h(c)_i,
\dots_{1\le i<n}
} +_\elt J.\]

\begin{lemma}
\label{kernel}
$\ker{\Phi}=0\oplus \ker{H}$.
\end{lemma}
\begin{proof}
Suppose $\Phi=0$. Indexing the elements of $\Phi$ as $0,\dots,n-1$, we have
\[D(e_{0,1})\Phi_0+\sum_{1\le i<n} D(e_{i,0})\Phi_i=0\]
which is equivalent to
\[D(e_{0,1})\ell(u)-\sum_{1\le i<n} D(e_{i,0})g(u_i)\in D(e_{0,1})J-\sum_{1\le i<n} D(e_{i,0}) J;\]
the $h(c)$ terms cancel out.

The left hand side simplifies to $\sum_{1\le i<n} D(e_{i,0})D(e_{0,0})\B(e_{0,1}u_i)$.
Factoring out $D(e_{0,0})$ and left-multiplying by $W$ on both sides yields
\[\sum_{1\le i<n} e_{i,0}(e_{0,1}u_i) \in e_{0,1}U-\sum_{1\le i<n} \cancel{e_{i,0}U}\]
\[\Rightarrow \sum_{1\le i<n} e_{i,1}u_i \in e_{0,1}U\]
\[\Rightarrow u=0.\]

Thus, $\Phi=(0,h(c))+J$, forcing $H(c)=0$.
\end{proof}

By comparing dimensions, Lemma \ref{kernel} implies
\[(n-1)\dim{U}+\dim{U^{\wedge (n-1)}}-\ker{H}
\le \rk{L}+(n-1)\dim{\ker{W}/J}.\]
By the rank-nullity theorem, $\dim{U^{\wedge (n-1)}}-\ker{H}=\rk{H}$.
Simplifying and substituting the known values of $\dim{U}$ and $\dim{\ker{W}/J}$ yields
\begin{equation}
\label{bound-a}
(n-1)\dim{U}+\rk{H}
\le \rk{L}+(n-1)\dim{\ker{W}/J}
\end{equation}
\[(n-1)\paren{\dim{\ker{W}/J}-\dim{U}}\ge \rk{H}-\rk{L}\]
\[(n-1)(r-(n+2)m)\ge \rk{H}-\rk{L}\]
\begin{equation}
\label{bound-b}
r\ge (n+2)m+\frac{\rk{H}-\rk{L}}{n-1}.  
\end{equation}

We first establish a lower bound on $\rk{H}$ using a general property of skew augmentations:
\begin{lemma}
\label{skew-augment}
Let $V$ and $W$ be vector spaces, $d$ be a positive integer, and $\ell:V^d\to W$ be a linear map.
Define $h:V^{\wedge d}\to W^d,\ h(c):=\paren{\ell(c_{1,:}),\dots,\ell(c_{d,:})}$.
Then $\rk{h}\ge \frac{d-1}{2}\rk{\ell}$.
\end{lemma}
\begin{proof}
Consider the map
\[\mu:(V^{\wedge d})^2\to (V^d)^{d-1};\ 
\mu(c,c'):=\paren{c_{1,:}+c_{2,:},\dots,c_{d-1,:}+c'_{d,:}}.\]

We prove $\mu$ is injective. Suppose $\mu=0$:
then $c_{i,j}+c'_{i+1,j} \ \forall 1\le i<d,\ 1\le j\le d$.
We know $c_{i,i}=0$ and $c'_{i,i}=0$,
so substituting $j=i$ and $j=i+1$ forces $c'_{i+1,i}=0$ and $c_{i,i+1}=0$, respectively.
By antisymmetry, we also have $c_{i+1,i}=0$ and $c'_{i,i+1}=0$.
Then substituting $j=i-1, i+2$ forces $c_{i,i+2}, c'_{i,i+2}=0$,
and repeating this process forces $c,c'=0$.

Since the input and output spaces of $\mu$ have the same number of dimensions, $d(d-1)\dim{V}$, $\mu$ is also surjective.

Now observe that
\[(\dots,h(c)_i+h(c')_{i+1},\dots_{1\le i<d})=\ell_\elt(\mu(c,c')),\]
so applying a fixed linear transformation on $(h(c),h(c'))$ results in a output that takes all possible values of $\ell(V^d)^{d-1}$ as $(c,c')$ is varied. Thus, $2\rk{h}\ge (d-1)\rk{\ell}$.
\end{proof}

Applying Lemma \ref{skew-augment} on $H,L$ with $d=n-1$ gives $\rk{H}\ge \frac{n-2}{2}\rk{L}$. Since it is possible that $\frac{\rk{H}-\rk{L}}{n-1}<0$ for $n=3$, we analyze that case separately from other values of $n$.

\subsection{$n=3$}
Applying the inequality $\rk{H}\ge \frac{1}{2}\rk{L}$ directly to Eq. \ref{bound-a} would yield
\[(n-1)\dim{U}\le \underbrace{\rk{L}-\rk{H}}_{\le \frac{1}{2}\rk{L}}+(n-1)\dim{\ker{W}/J}\]
\[(n-1)\dim{U}\le \frac{1}{2}\underbrace{\rk{L}}_{\le \dim{\ker{W}/J}}+(n-1)\dim{\ker{W}/J}\]
\[(n-1)\dim{U}\le \paren{n-\frac{1}{2}}\dim{\ker{W}/J}\]
\[2\dim{U}\le \frac{5}{2}\dim{\ker{W}/J}\]
\[2m\le \frac{5}{2}(r-4m)\]
\[r\ge \frac{24m}{5}.\]

To improve this bound, define
\begin{itemize}
    \item $\Delta=2\rk{H}-\rk{L}$
    
    \item $N=\ker{H}\subseteq U$

    Note that $h(x)=(-f_2(x), f_1(x))$, so $N=\ker{f_1+J}\cap \ker{f_2+J}$.
    
    \item $S=\brace{x\in N: f_1(x), f_2(x)\in f_0(N)}$
\end{itemize}

\begin{lemma}
\label{32m-helper-subspace}
$\dim{S}\ge \dim{N}-\Delta$.
\end{lemma}
\begin{proof}
Suppose $H(x)=0$, i.e. $h(x)\in J$.
By definition, there exist $u_1,u_2\in U$ such that $h(x)=(f_0(u_1),f_0(u_2))$.
Since $f_0$ is injective, $u_1$ and $u_2$ are in fact linearly related to $x$, i.e. there exist fixed maps $\sigma_1,\sigma_2:N\to \F^{2\times m}$ such that $u_1=\sigma_1(x), \ u_2=\sigma_2(x)$.

Since $h(x)$ satisfies $\sum_{1\le i<n} D(e_{i,0})h(x)_i=0$ for arbitrary $x$,
we immediately have $D(e_{1,0})f_0(u_1)+D(e_{2,0})f_0(u_2)=0$.
Multiplying by $D(e_{0,0})^{-1}$ yields $f_1(u_1)+f_2(u_2)=\ell(u_1,u_2)=0$, thus $L(u_1,u_2)=0$.

Under the condition $x\in N$, $x$ is in $S$ if and only if $(u_1,u_2)\in N^2\subseteq \ker{L}$.
Thus, forcing $x$ be in $S$ imposes at most $\dim{\ker{L}}-2\dim{N}$ extra linear constraints on $x$.
Since
\[\dim{\ker{L}}-2\dim{N}
=\paren{\cancel{\dim{U^2}}-\rk{L}}-2\paren{\cancel{\dim{U}}-\rk{H}}
=2\rk{H}-\rk{L}
=\Delta,\]
we have $\dim{S}\ge \dim{N}-\Delta$.
\end{proof}

If $\dim{S}=0$, $\Delta\ge \dim{N}$.
Otherwise, since $f_0(S)+f_1(S)+f_2(S)\subseteq f_0(N)$, applying Lemma \ref{subspace-expand} from strong normalization gives $\dim{S}+2\le \dim{N}$, forcing $\Delta\ge 2$. Combining these cases gives $\Delta\ge \min(2,\dim{N})$.

If $\Delta\ge 2$, then $\rk{H}\ge \frac{1}{2}\rk{L}+1$, so
\[(n-1)\dim{U}\le \underbrace{\rk{L}-\rk{H}}_{\le \frac{1}{2}\rk{L}-1}+(n-1)\dim{\ker{W}/J}\]
\[(n-1)\dim{U}+1\le \frac{1}{2}\underbrace{\rk{L}}_{\le \dim{\ker{W}/J}}+(n-1)\dim{\ker{W}/J}\]
\[(n-1)\dim{U}+1\le \paren{n-\frac{1}{2}}\dim{\ker{W}/J}\]
\[2\dim{U}+1\le \frac{5}{2}\dim{\ker{W}/J}\]
\[2m+1\le \frac{5}{2}(r-4m)\]
\[r\ge \frac{24m+2}{5}.\]

Otherwise, $\dim{N}\le 2$, so $\rk{H}\ge m-2$ by the rank-nullity theorem,
and
\[(n-1)\dim{U}+\rk{H}\le \rk{L}+(n-1)\dim{\ker{W}/J}\]
\[(n-1)\dim{U}+(m-2)\le n\dim{\ker{W}/J}\]
\[3m-2\le 3(r-4m)\]
\[r\ge 5m-\frac{2}{3}\]
\[r\ge 5m \ \textrm{ by integrality}.\]

\subsection{$n\ge 4$}
Lemma \ref{skew-augment} gives $\rk{H}\ge \frac{n-2}{2}\rk{L}$
and, and Lemma \ref{subspace-expand} (strong normalization) gives $\rk{L}\ge n-1$.
Applying these inequalities directly to Eq. \ref{bound-b} would yield
\[r\ge (n+2)m+\frac{n-4}{2}.\]

To get a stronger bound, we utilize more properties of the map $H$.
Define
\begin{itemize}
    \item $F_i(x)=f_i(x)+J$
    \item $\rho=\max_{\gamma\in\F^{n-1}} \rk{\sum_i \gamma_iF_i}$
    \item $W:=\brace{[c_{i,j}]_{1\le i,j<n}\in U^{\wedge (n-1)}:\ c_{i,:}=0 \ \forall i>1}$
\end{itemize}

\begin{lemma}
\label{n2m-help}
$\rk{H}\ge \frac{n-1}{2}\rho+\frac{n-3}{2}\rk{L}$.
\end{lemma}
\begin{proof}
Without loss of generality we can assume $\rho=\rk{F_1}$, since for any invertible $\xi\in \F^{(n-1)\times (n-1)}$ we can define
\[L'(u):=L(u\xi),\ u\in U^{n-1};\
H'(c):=[\dots,L'(c_{i,:}),\dots_i],\ c\in U^{\wedge(n-1)};\]
where we treat $u$ and $c$ as a vector and matrix, respectively, of vectors, i.e. $u\xi=[\sum_i u_i \xi_{i,j}]_j$.
Then $H'$ is the skew augmentation of $L'$,
and $\rk{H}=\rk{H'}$ because $H(\xi^\intercal c\xi)=\xi^\intercal H'(c)$ and the map $c\mapsto \xi^\intercal c \xi$ is a bijection from $U^{\wedge (n-1)}$ to itself.

Define $\Gamma:=\brace{[c_{i,j}]_{1\le i,j<n}\in U^{\wedge (n-1)}:\ c_{i,:}=0 \ \forall i\ge 2}$
and $\underline{H}(c):=[H(c+\Gamma)_i]_{2\le i<n}$, where we remove the first element of $H(c+\Gamma)$.
Then $\underline{H}$ is the skew augmentation of $\underline{L}(u_2,\dots,u_{n-1}):=L(U,u_2,\dots,u_{n-1})
$,
so $\rk{\underline{H}}\ge \frac{n-3}{2}\rk{\underline{L}}$.
We also have $\rk{\underline{L}}=\rk{L}-\rk{F_1}=\rk{L}-\rho$.

Additionally, $H(\Gamma)\ge (n-2)\rho$, since for any $c\in \Gamma$ we have $H(c)=(0,F_1(c_{2,1}),\dots,F_1(c_{n-1,1}))$, which contains $n-2$ independent outputs of $F_1$.

Thus, $\rk{H}\ge \dim{H(\Gamma)}+\rk{\underline{H}}\ge (n-2)\rho+\frac{n-3}{2}\rk{\underline{L}}
=\frac{n-1}{2}\rho+\frac{n-3}{2}\rk{L}$.
\end{proof}

\begin{lemma}
If $m\ge n$, then $\rho\ge n-1$.
\end{lemma}
\begin{proof}
Since $\dim{J}=m\ge n-1$, there exists a subset $I\subseteq \brace{0,\dots,r-1},\ |I|=n-1$ such that restricting $J$ at indices $I$ covers all of $\F^I$.
Then since $n>|I|$, there exists some nonzero matrix $M$ that is a linear combination of $D(e_{0,0}),\dots,D(e_{n-1,0})$ such that $M_{I,I}=0$.

By Corollary \ref{lin-combo-injective}, the map $f':U\to\F^r,\ f'(u)=M\B(u)$ is injective. Since $f'$ vanishes at indices $I$, $\dim{f'(U)+J}\ge m+n-1$: $f'(U)$ contributes $m$ dimensions and $J$ contributes $n-1$.

Thus, the quotient map $F'(u)=f'(u)+J$ has rank $\ge (m+n-1)-\dim{J}=n-1$. But $F'$ is a linear combination of $F_1,\dots,F_{n-1}$, so $\rho\ge n-1$.
\end{proof}

Applying the previous lemmas to Eq. \ref{bound-b},
\[r\ge (n+2)m+\frac{1}{n-1}\paren{\rk{H}-\rk{L}} \ \textrm{ if } m\ge n\]
\[r\ge (n+2)m+\frac{1}{n-1}\paren{\frac{n-1}{2}\rho+\frac{n-5}{2}\rk{L}}.\]

Since $\rk{L}\ge n-1$ by strong normalization, for $n\ge 5$ the bound simplifies to
\[r\ge (n+2)m+\paren{\frac{n-1}{2}+\frac{n-5}{2}}\]
\[r\ge (n+2)m+\paren{n-3}.\]

\subsubsection{$n=4$}
The case $n=4$ must be handled separately because the $\frac{n-5}{2}$ coefficient is negative.
In this case, our goal is to prove $r\ge 6m+1$.
We already know $r\ge 6m$ due to Eq. \ref{bound-b} and the skew-augmentation bound $\rk{H}\ge \frac{4-2}{2}\rk{L}=\rk{L}$.

We make one final observation:
\begin{lemma}
\label{rank-tightness}
$\rk{L}\le (n-1)\rho$. Furthermore, equality forces $\rk{L}=(n-1)m$ if $n\ge 3$.
\end{lemma}
\begin{proof}
The first inequality is clear, since $(n-1)\rho\ge \rk{F_1}+\dots+\rk{F_{n-1}}\ge \rk{L}$.

If $\rk{L}=(n-1)\rho$, then $\rk{L}=\rk{F_1}+\dots+\rk{F_{n-1}}$, forcing the images of the $F_i$ to form a direct sum and for all $\rk{F_i}$ to equal $\rho$.
Then $S:=\sum_{1\le i<n} F_i$ is isomorphic to the map $x\mapsto (F_1(x),\dots,F_{n-1}(x))$, so
$\rk{S}=m-\dim{N}$, where $N:=\bigcap_i \ker{F_i}$.

By definition, $\rk{S}\le \rho$, so $\dim{N}\ge m-\rho$;
but since all $\rk{F_i}=\rho$, $\dim{N}=\dim{\ker{F_i}}$, so all $N=\ker{F_i}$ for all $i$.

We now show that $f_i(N)\subseteq f_0(N)$ for all $i\ne 0$. For any $v\in N$, there exist $u_i\in U$ such that $f_i(v)=f_0(u_i)$. Then $\B(u_i)=D(e_{0,0})^{-1} D(e_{i,0})\B(v)$, so for any two indices $i,j\ne 0$ we have $D(e_{j,0})\B(u_i)=D(e_{i,0})\B(u_j)
\Rightarrow f_j(u_i)=f_i(u_j)
\Rightarrow F_j(u_i)=F_i(u_j)$. Since the $F_i$ form a direct sum, $\im{F_i}\cap \im{F_j}=0$, so $u_i,u_j\in N$.

Thus, applying Lemma \ref{invariant-subspace} forces $N=0$, so $\rho=m$.
\end{proof}

Suppose $r=6m$: then Eq. \ref{bound-b} forces $\rk{H}-\rk{L}\le 0$, which is tight with the skew-augmentation bound and thus forces $\rk{H}=\rk{L}$.
By Lemma \ref{n2m-help}, $\rk{L}\ge \frac{n-1}{2}\rho+\frac{n-3}{2}\rk{L}=\frac{3}{2}\rho+\frac{1}{2}\rk{L}$,
implying $\rk{L}\ge 3\rho=(n-1)\rho\ge \rk{L}$, forcing $\rk{L}=3\rho$. By Lemma \ref{rank-tightness}, $\rk{L}=(n-1)m=3m$. But $\rk{L}\le \dim{\ker{W}/J}=r-5m=m$: contradiction.


\begin{thebibliography}{}
\bibitem{blaser99}
Markus Bl\"aser.
Lower bounds for the multiplicative complexity of matrix multiplication. computational complexity, 8(3):203-26,
1999.


\bibitem{survey}
Markus Bl\"aser.
Fast matrix multiplication.
Theory of Computing Library, 5:1-60,
2013.

\bibitem{proof323}
Vladimir Petrovich Burichenko,
On bilinear complexity of multiplication of a $3\times 2$ matrix by a $2\times 3$ matrix.
Diskretnaya Matematika, 36(1):15–45,
2024.

\bibitem{hopcroft}
J. E. Hopcroft, L. R. Kerr.
On minimizing the number of multiplications necessary for matrix multiplication.
SIAM Journal on Applied Mathematics, 20(1):30-6,
1971.

\bibitem{novikov}
Alexander Novikov, Ngân Vũ, Marvin Eisenberger, Emilien Dupont, Po-Sen Huang, Adam Zsolt Wagner, Sergey Shirobokov, Borislav Kozlovskii, Francisco J. R. Ruiz, Abbas Mehrabian, M. Pawan Kumar, Abigail See, Swarat Chaudhuri, George Holland, Alex Davies, Sebastian Nowozin, Pushmeet Kohli, Matej Balog.
AlphaEvolve: A coding agent for scientific and algorithmic discovery.
2025.
\url{https://arxiv.org/abs/2506.13131}

\bibitem{shitov}
Yaroslav Shitov.
Determining the tensor rank of $(2, 2, M)$ matrix multiplication with GPT-5.6 Sol.
2026.
\url{https://doi.org/10.13140/RG.2.2.29432.61441}

\bibitem{tsyganov}
Askar Tsyganov, Uliana Parkina, Sergey Samsonov, Maxim Rakhuba
A lower bound for $\ang{3,2,m}$ matrix multiplication.
2026.
\url{https://arxiv.org/abs/2609.22054}
\end{thebibliography}
\end{document}